%% file: arxiv-netrunner.tex
\documentclass[12pt]{article}
\pdfoutput=1

\usepackage{times}
\usepackage{amsthm}
\usepackage{amsmath}
\usepackage{amssymb}
\usepackage{mathabx}
\usepackage{amsfonts}
\usepackage{graphicx}
\newtheorem{theorem}{Theorem}[section]

\usepackage[svgnames]{xcolor}
\usepackage{floatpag}
\usepackage{makecell}
\usepackage[raster]{tcolorbox}
\usepackage{microtype}
\usepackage[
  colorlinks=false,
  urlbordercolor=blue,
  pdfborderstyle={/S/U/W 1},
  pdftitle={Netrunner Mate-in-1 or -2 is Weakly NP-Hard},
  pdfauthor={Jeffrey Bosboom and Michael Hoffmann},
  pdflang=en,
]{hyperref}

\definecolor{xxxcolor}{rgb}{0.8,0,0}

\usepackage{xparse}
\ExplSyntaxOn
\DeclareExpandableDocumentCommand{\myrepeat}{O{}mm}
 {
  \int_compare:nT { #2 > 0 }
   {
    #3 \prg_replicate:nn { #2 - 1 } { #1#3 }
   }
 }
\ExplSyntaxOff

\newcommand{\ccP}{\textrm{\textsc{P}}}

\newcommand{\ccNP}{\textrm{\textsc{NP}}}

\title{Netrunner Mate-in-1 or -2 is Weakly \ccNP-Hard}
\author{
  Jeffrey Bosboom\\
  MIT CSAIL\\[-.25em]
  \small{\url{jbosboom@csail.mit.edu}}
  \and
  Michael Hoffmann\\
  ETH Zurich\\[-.25em]
  \small{\url{hoffmann@inf.ethz.ch}}
}
\date{}

\begin{document}

\maketitle

\begin{abstract}
We prove that deciding whether the Runner can win this turn (mate-in-1) in the Netrunner card game generalized to allow decks to contain an arbitrary number of copies of a card is weakly \ccNP-hard.  We also prove that deciding whether the Corp can win within two turns (mate-in-2) in this generalized Netrunner is weakly \ccNP-hard.
\end{abstract}

\section{Introduction}

Netrunner is a Gibsonian-cyberpunk-themed asymmetric collectible card game for two players, the Runner and the Corporation.
Both sides can play cards that either provide a one-shot effect or stay in play to provide an ongoing effect or available resource. The Corporation attempts to advance agendas stored in servers protected by ice (Intrusion Countermeasures Electronics), while the Runner makes runs on those servers to steal the agendas. 
There are two versions of the game: the original 1996 version\footnote{see \href{https://www.boardgamegeek.com/boardgame/1301/netrunner}{the corresponding BoardGameGeek page} for information, card list at \url{http://emergencyshutdown.net/webminster}}
and a 2012 reboot named \emph{Android: Netrunner}\footnote{Official website (including rules/FAQ download) at \url{https://www.fantasyflightgames.com/en/products/android-netrunner-the-card-game/}, card list at \url{https://netrunnerdb.com/en/search}}.  The rules of the games are broadly similar and many cards from the original game have been reprinted in the new game (sometimes under new names or with different costs).  We study \emph{Android: Netrunner} in this paper, proving that the mate-in-1 problem for the Runner (deciding whether the Runner can win this turn) and the mate-in-2 problem for the Corp are weakly \ccNP-hard when the deckbuilding rules are relaxed to remove the limit of three copies of a card per deck.

\section{Related Work}

For most two-player games, deciding which player will win from a given board state (generalized to arbitrary board sizes) under optimal play is PSPACE- or EXPTIME-complete, depending on whether the game prohibits or allows loops.  (Appendix A of \cite{GPC} lists some examples, among many results on games and puzzles.)  For only a few games is the mate-in-1 problem (deciding if the current player can win with the next move) interesting to study, because most games have only polynomially many legal moves from each position and whether the player won with a move is easy to verify.  We are aware of two nontrivial mate-in-1 results: Conway's Phutball is an example of a game for which mate-in-1 is \ccNP-complete~\cite{Phutball}, while mate-in-1 for Checkers is in \ccP~\cite{Checkers} (but not obviously so).

Four-player Magic:~the Gathering in which all players always take ``may'' actions when possible is Turing-complete \cite{MtGTuring}.  The game ends if and only if the simulated Turing machine halts, but the previous player wins, so this proof is a lose-in-1 result, albeit with a turn of possibly unbounded length.

\section{Netrunner}

Netrunner restricts both players to use no more than three copies of each card and does not provide for any card substitutes (like Magic's tokens) to be created during play. While there are cards that can accumulate counters during play, with only a finite number of cards available it does not seem possible to create arbitrarily large problem instances. Thus we relax the deckbuilding constraints to allow any number of copies of a card to be in a deck. We call the resulting game \emph{card-copy-generalized Netrunner}.

In the following proofs, we assume the reader is familiar with the Netrunner rules, though the concept of the reduction should be clear to all.  For convenience, the text of the cards used in the reduction is provided in Appendix~\ref{sec:netrunner-cards1} and \ref{sec:netrunner-cards2}.

\begin{theorem}
Deciding whether the Runner can win this turn in card-copy-gener\-alized Netrunner is weakly \ccNP-hard.
\end{theorem}

\begin{proof}
  We reduce from \textsc{2-Partition}~\cite{GareyJohnson}: Given a multiset $A$ of positive integers, is there a partition $A=A_1\uplus A_2$ such that $|A_1| = |A_2|$ and the sum of the integers in $A_1$ is equal to the sum of the integers in $A_2$?

 Consider an input instance $A=\{a_1,\ldots,a_n\}$ and let $t = \sum_{i=1}^{n} a_i / 2$ denote the target sum.  We represent each $a_i$ by an Ice Wall with strength $3(a_i - 1) + 1$, which costs $2a_i$ to break with Aurora.  The Runner solves the \textsc{2-Partition} instance by rearranging the ice with Escher, then makes runs that verify the solution.  The rest of the infrastructure exists to ensure the Runner cannot win this turn without solving the \textsc{2-Partition} instance.


\textbf{Corp board state (Figure~\ref{fig:netrunner-corp-board}):} The Corp's identity is Weyland Consortium: Building a Better World\footnote{We choose these Runner and Corp identities because their abilities do not affect play on the turn being analyzed.  Because we are not generalizing influence, the Corp must play a Weyland identity to include an arbitrary number of Ice Walls in their deck.}.  R\&D is protected by $|A| + 3$ Enigmas and one of the Ice Walls in the outermost position; the contents of R\&D are irrelevant because we will show the Runner cannot access them.  There is one card stored in HQ, Priority Requisition, and two Strongboxes installed in the root of HQ.  HQ is protected by $(|A|-2)/2$ of the Ice Walls and an Archer in the outermost position.  Archives does not contain any agendas (all cards are face up), but it does contain a copy of Fast Track that was played last turn.  Archives is not protected by any ice.  The Corp has two remote servers.  One remote server has an installed Dedicated Response Team protected by $|A| + 3$ Enigmas and one of the Ice Walls in the outermost position.  The other remote server has an installed Priority Requisition (exposed last turn when the Runner played Infiltration) and K.~P.~Lynn, protected by the remaining $(|A|-2)/2$ Ice Walls and an Archer in the outermost position.  All assets, upgrades and ice are rezzed.

\begin{figure}
\vspace{-.5in}
\centerline{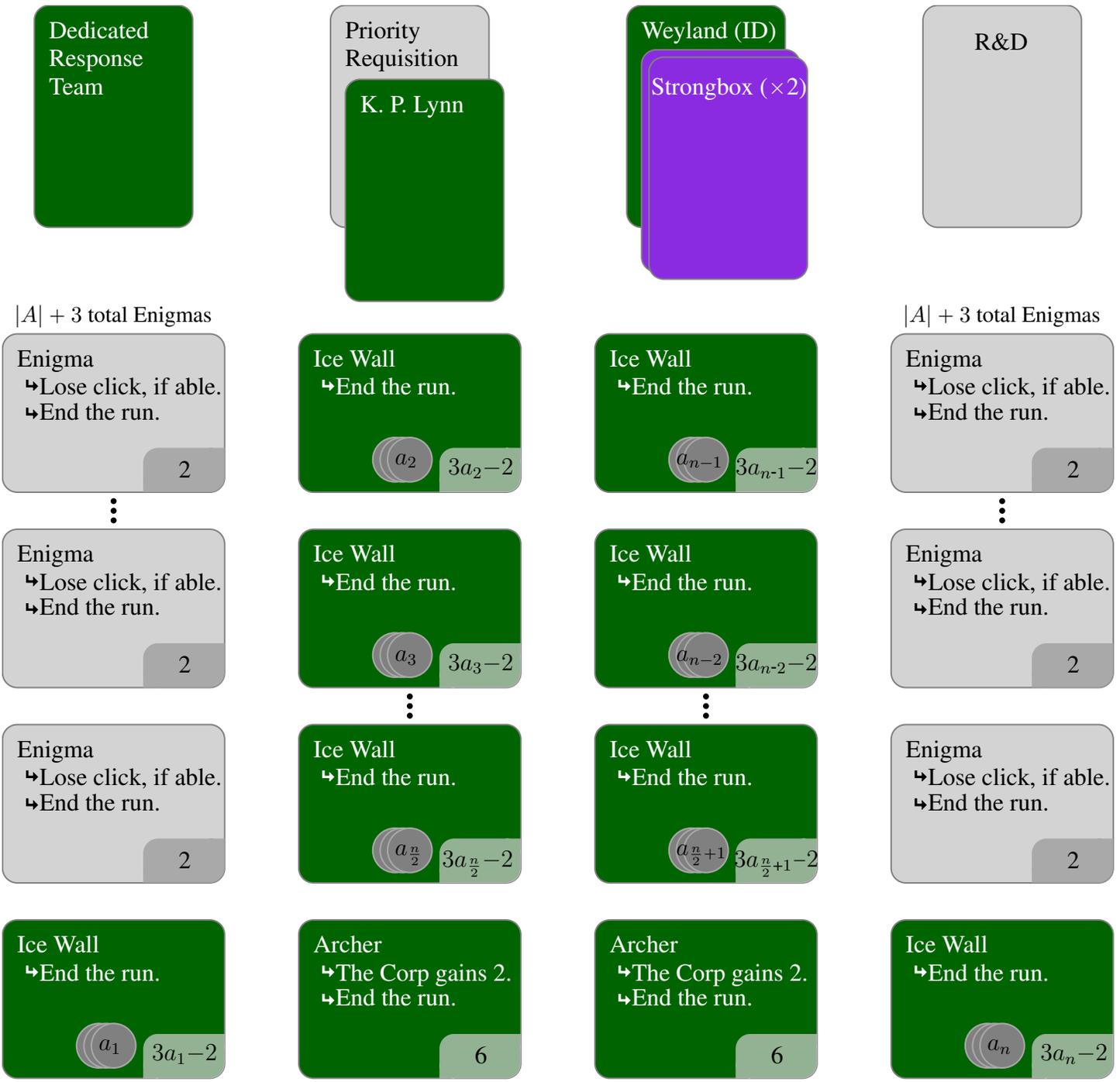}
\caption{The Corp's board state produced by the reduction.  To save space, Archives (no ice) and the Priority Requisition held in HQ are not shown.  The ice on R\&D and the Dedicated Response Team remote is more than twice as deep as on the other servers. 
}
\label{fig:netrunner-corp-board}
\thisfloatpagestyle{empty}
\end{figure}

\textbf{Runner board state:} The Runner's identity is Exile: Streethawk.  The Runner's rig consists of Aurora, Grappling Hook, and Pheromones with $4t + c + 2$ virus counters and credits on it, where $c$ is the total cost to break the Ice Walls initially protecting HQ with Aurora.  The Runner has Escher in their grip and $2t + 3$ credits in their credit pool.  The stack is empty.  The contents of the heap are not relevant to play on this turn because neither player has any card abilities that interact with the heap, but the heap contains one copy of Infiltration that was played last turn.  The Runner has a stolen Priority Requisition in their score area (worth 3 agenda points out of the 7 points required to win).

\textbf{Line of play:} To win this turn, the Runner must steal at least 4 points of agendas.  The Runner does not know the top card of R\&D, but the Runner cannot break the Enigmas protecting R\&D anyway.  There are $|A| + 3$ Enigmas, but only $|A| + 2$ non-Enigma pieces of ice that Enigmas could be exchanged with, so by the pigeonhole principle there will always be an Enigma protecting R\&D even after the Runner plays Escher.  All cards in Archives are faceup, so the Runner can see there are no agendas in Archives.  Thus, to win this turn, the Runner must steal the Priority Requisitions in the remote (revealed last turn by Infiltration) and in HQ (revealed last turn by Fast Track).

To steal the Priority Requisition from the remote, the Runner must take a tag from K.~P.~Lynn.  If the Runner doesn't immediately win by stealing the Priority Requisition, Dedicated Response Team will flatline them when the run ends because they are tagged.  By the same pigeonhole argument that makes R\&D inaccessible, the Runner cannot trash Dedicated Response Team, so the Runner must run the Priority Requisition remote after stealing the other Priority Requisition from HQ.

The Runner cannot break Archer's subroutines except once via Grappling Hook, so the Runner can only pass Archer once.  Both agendas are in servers protected by an Archer, so the Runner must play Escher to move the Archers.  The Runner cannot break Enigma, so the Runner must swap the Archers with the two Ice Walls on R\&D and the other remote.  The Runner cannot afford to pay the additional click cost imposed by the Strongboxes, because one click spent playing Escher, one click spent to make the run, and two extra clicks leaves the Runner with no clicks remaining to run the Priority Requisition remote.  Thus the Runner must run HQ at least twice after Escher, once to trash the Strongboxes and again to steal the agenda, then run the Priority Requsition remote.  The Runner has no clicks left over to click for credits.

The Runner spends 3 credits from their credit pool to play Escher, trashes Grappling Hook to pass Archer (allowing the Corp to gain two credits from Archer's first subroutine), then spends $c$ credits from Pheromones to break the other ice on HQ.  (Credits in the Runner's credit pool can be spent at any time, while credits on Pheromones can only be spent during runs on HQ, so it's to the Runner's advantage to spend credits from Pheromones when possible.)  When the Escher run ends, the Runner has three clicks remaining, $2t$ credits in their credit pool and $4t + 2$ credits on Pheromones\footnote{Pheromones gains a virus counter for each successful run on HQ, but recurring credits are only refreshed when a card is installed or rezzed (enters play) or at the beginning of its owner's turn, so the Runner does not get any additional credits to spend this turn.}.

During rearrangement, the Runner must arrange the Ice Walls such that the ice on HQ and the Priority Requisition remote each costs $2t$ credits to break.  If the Runner places more than $2t$ total cost to break on HQ, the Runner has too few credits on Pheromones to break all the subroutines on both runs and trash the Strongboxes.  The Runner cannot make up the deficit from their credit pool because they must run HQ twice, leaving them with too few credits to successfully run the Priority Requisition remote.  If the Runner places ice with total cost to break greater than $2t$ on the Priority Requisition remote, the Runner has too few credits in their credit pool to make a successful run.  Thus the Runner is forced to put $2t$ total cost to break on each server with an agenda in it.  This rearrangement of the ice represents a solution to the \textsc{2-Partition} instance, and the Runner can win this turn if and only if such a solution exists.
\end{proof}

\begin{theorem}
Deciding whether the Corp can win within two turns (mate-in-2) in card-copy-generalized Netrunner is weakly \ccNP-hard.
\end{theorem}

\begin{proof}
We again reduce from \textsc{2-Partition}.  Similarly to the Runner \mbox{mate-in-1} proof, this proof centers around an ice rearrangement effect, this time from Mandatory Seed Replacement, a Jinteki agenda.  Because agendas can only be used in-faction (in this case, with a Jinteki identity) and we are not generalizing influence, we cannot represent the $a_i$ with Ice Walls as in the previous proof.  Instead, we represent each $a_i$ as a Wall of Static with $a_i - 1$ Sub Boost counters, which costs Aurora $2a_i + 2$ credits to break (2 credits to boost Aurora's strength once, $2a_i$ credits to break the subroutines).  If the Corp places less than $2t + |A|$ cost to break on one of two remote servers, or tries to avoid simulating the \textsc{2-Partition} instance, the Runner can steal an agenda and win; otherwise, the Corp can finish advancing its agendas and win on its second turn.


\textbf{Corp board state:} The Corp's identity is Nisei Division: The Next Generation.  HQ contains a Medical Breakthrough.  R\&D contains two cards: a Medical Breakthrough (on top) and a Hedge Fund.  Archives does not contain any agendas (all cards are faceup).  R\&D, HQ, and Archives are all unprotected.  There are two remote servers.  One remote contains a Mandatory Seed Replacement with three advancement tokens, protected by half ($|A|/2$) of the Walls of Static representing $A$.  Nothing is installed in the other remote, which is protected by the other half of the Walls of Static.  The Corp has $4$ credits in their credit pool.  The Corp has a Priority Requisition in their score area (worth 3 agenda points out of the 7 points required to win).

\textbf{Runner board state:} The Runner's identity is Exile: Streethawk.  The Runner's rig consists of only Aurora.  The Runner's stack and grip are empty.  The contents of the heap are not relevant to play on this turn because the Runner has no card abilities that retrieve cards from their heap, but the heap contains one copy of The Shadow Net that was previously used to forfeit agendas.  The Runner has $2t + |A| - 4$ credits in their credit pool.  The Runner has a Priority Requisition and a Medical Breakthrough in their score area (worth 5 agenda points out of the 7 points required to win).

\textbf{Line of play:} The Corp's mandatory draw is the second Medical Breakthrough.  The Corp knows their decklist, so the Corp knows there are no more agendas in R\&D because, except for the Mandatory Seed Replacement currently installed, all other agendas have been scored, stolen, or forfeited to The Shadow Net.  Any winning line of play must include scoring one of the Medical Breakthroughs without allowing the Runner to steal one.  The Corp also knows the remaining card in R\&D is a Hedge Fund.

The Corp's first turn must be to advance and score the Mandatory Seed Replacement (bringing the Corp to 5 agenda points), then install both Medical Breakthroughs in the remote server.  If the Corp does not install both Medical Breakthroughs on its first turn, the Runner can win by stealing one from HQ (which is unprotected).  If the Corp installs a Medical Breakthrough in a new remote server, the Runner can win by running the new server, as the Corp has no ice to protect it with\footnote{While Mandatory Seed Replacement does not contain the sentence on Escher requiring that the same number of ice must protect each server after the rearrangement, the Unofficial FAQ (\url{https://netrunnerdb.com/en/set/fm/rulings}) clarifies it is still required.  Despite the name, the Unofficial FAQ has the approval of the lead designer, Michael Boggs, so it is effectively official.}.  If the Corp installs a Medical Breakthrough over the Mandatory Seed Replacement (trashing it), the Runner can win next turn by running Archives and stealing the Mandatory Seed Replacement.  If the Corp draws the Hedge Fund (to make HQ accesses probabilistic or to play it), they will lose when they cannot take their mandatory draw at the beginning of their next turn (and there is no way to win this turn).  Clicking for credits or playing the Hedge Fund are pointless as the Corp has nothing to spend the extra credits on.

If the Corp places less than $2t + |A|$ cost to break on either remote server when rearranging the ice, the Runner can win on their turn by spending three clicks to gain three credits (for a total of $2t + |A| - 1$ credits available), then stealing the Medical Breakthrough from that remote server.  If the Corp places exactly $2t + |A|$ cost to break on both remote servers (i.e., solves the \textsc{2-Partition} instance), the Runner does not have enough credits to successfully run either server.  Then on the Corp's second turn, they spend all three clicks and all $3$ credits remaining in their credit pool to advance one of the Medical Breakthroughs.  The Medical Breakthrough in the Runner's score area lowers the advancement requirement of the installed Medical Breakthrough to 3, so the Corp can score it, reaching the 7 points needed to win.  Thus the Corp can win within two turns if and only if the \textsc{2-Partition} instance has a solution.
\end{proof}

\section{Discussion and Open Problems}

Mandatory Seed Replacement will eventually rotate out of tournament play.  It is open whether the proof can be adapted to use cards only from the ``big box'' sets (which will never rotate).

In the Runner proof, Pheromones serves as a source of credits that can be spent only on one server.  Before the announcement of the Revised Core Set (without which Pheromones would have rotated), we developed an alternative proof based on Datasucker and Navi Mumbai City Grid.  Datasucker virus counters can be used in lieu of spending credits to boost icebreaker strength, but not on servers with a rezzed Navi Mumbai City Grid.  Our alternative proof used Corroder (removed from the Revised Core Set), so a virus counter directly offsets a credit that would have been spent to boost Corroder's strength; some care would be required to adapt for other icebreakers.  We mention this alternative construction because it may be useful in other proofs, or in case Fantasy Flight Games changes their rotation policy again.

The Corp proof is about mate-in-2, so it is natural to ask if mate-in-1 remains hard.  One proof strategy is to combine Mandatory Seed Replacement with An Offer You Can't Refuse.  If the Corp plays An Offer You Can't Refuse while already at 6 agenda points, the Runner is forced to run.  If the Corp solves the \textsc{2-Partition} instance when rearranging the ice, the Runner does not have enough credits to break all the ice and will be flatlined by a subroutine; otherwise, the Runner has enough credits to survive.  Unfortunately, the Corp chooses the server for An Offer You Can't Refuse, so they can ensure a win by placing the ice costing the most to break on that server.  Ice with the \textbf{deflector} subtype can redirect the Runner to another server, giving the Runner a choice of servers, but we do not know how to force the Corp to place the deflector ice in the outermost position when rearranging.

While it seems likely that the problems we consider are contained in NP, it is not obvious how to prove there is no combination of card effects allowing more than a polynomial number of actions within a turn.  One strategy for such a proof is to state a potential function whose parameters are the resources available to the players (clicks, credits in the credit pools, credits and counters on cards, cards remaining in the stack/R\&D, ...) and show every action decreases it, but this would require a case analysis of all cards in the game to ensure there are no opportunities for resource arbitrage.

The complexity of deciding the winner of a Netrunner game under perfect play remains open.


\section*{Acknowledgments}

We would like to thank the participants and organizers of the 30th Bellairs Winter Workshop on Computational Geometry.

\appendix
\newpage\section{Netrunner Card Details: Corp}\label{sec:netrunner-cards1}

\newcommand{\subroutine}{\includegraphics[height=\fontcharht\font`X]{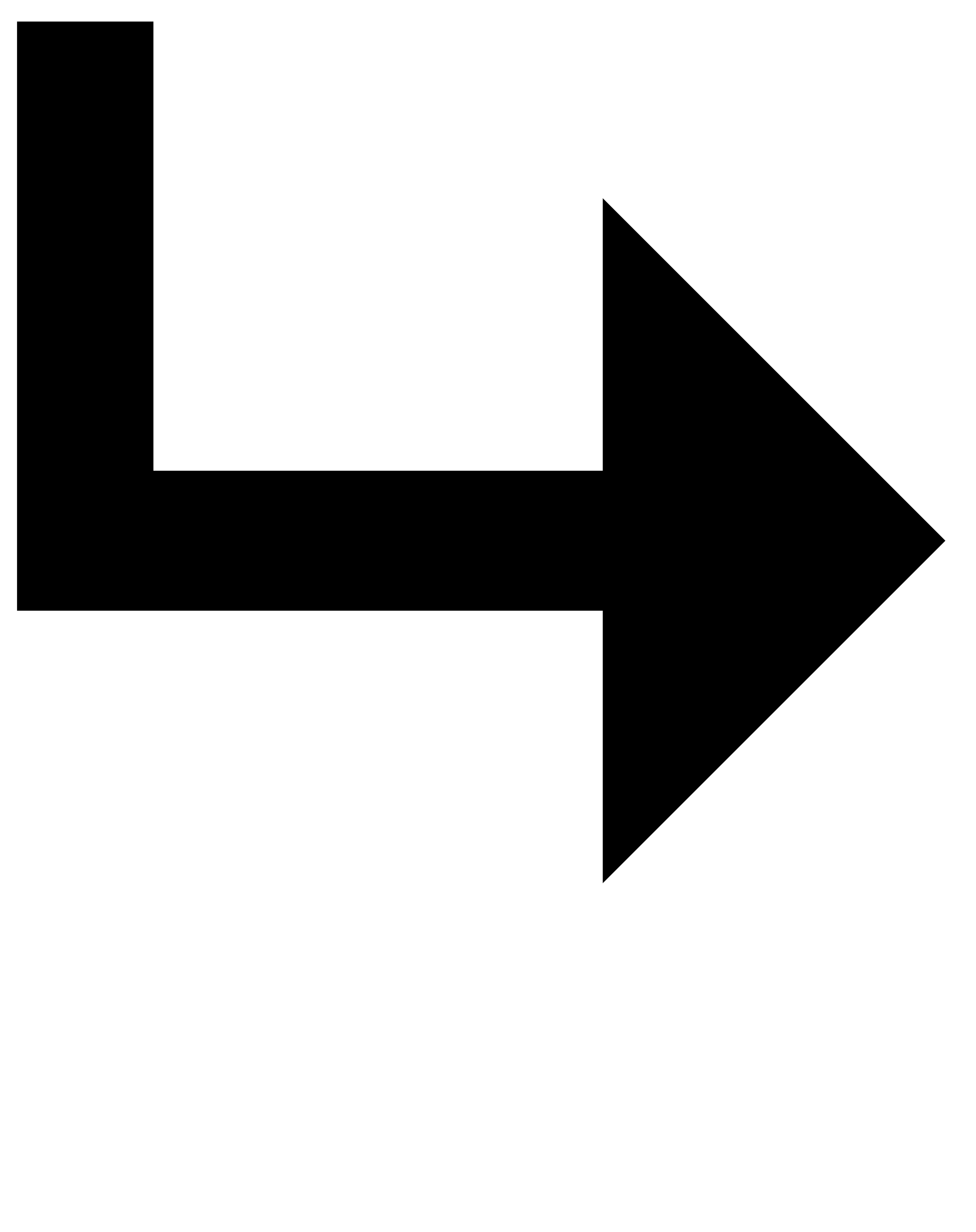}}
\newcommand{\credit}{\includegraphics[height=\fontcharht\font`X]{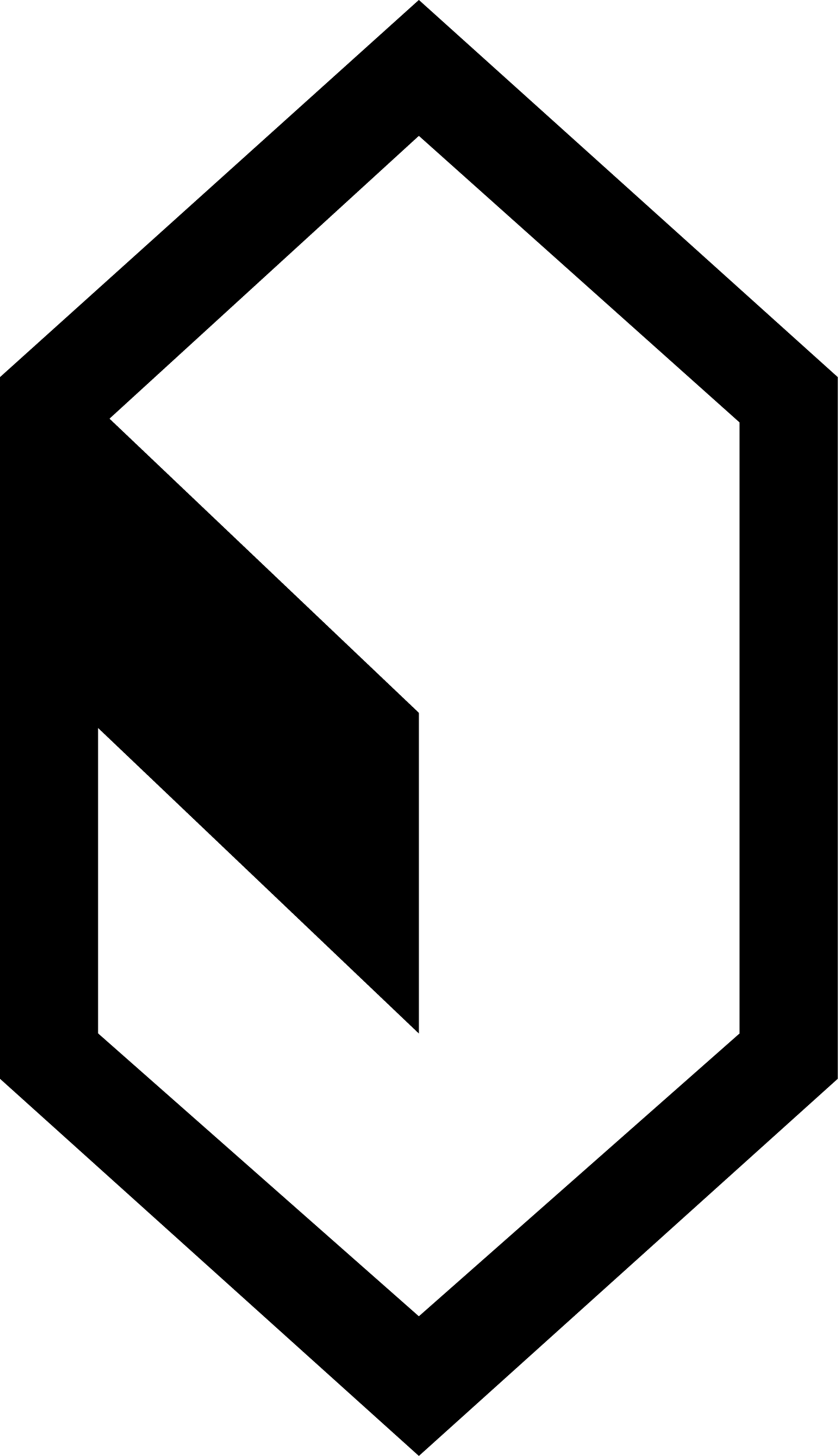}}
\newcommand{\recurcredit}{\includegraphics[height=\fontcharht\font`X]{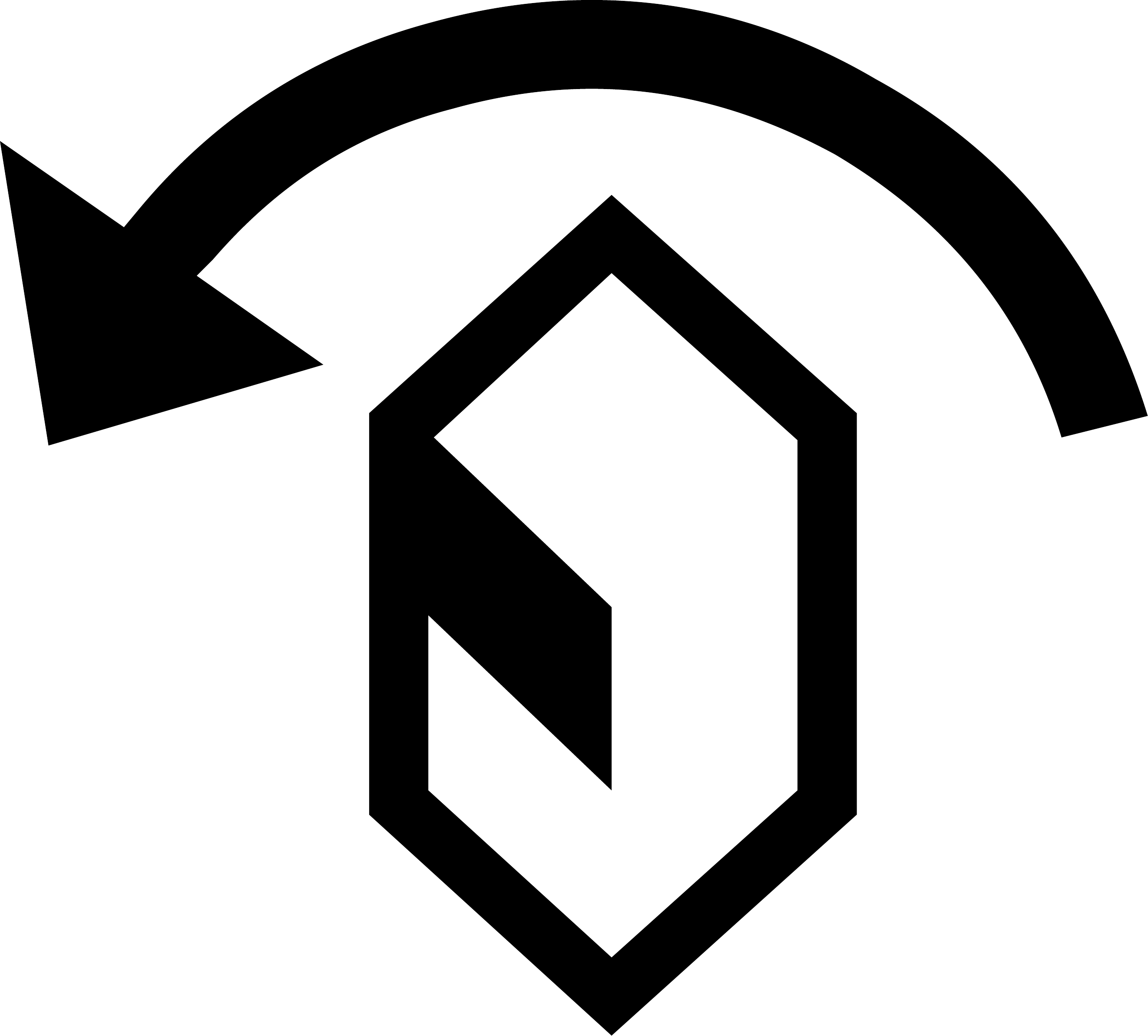}}
\newcommand{\click}{\includegraphics[height=\fontcharht\font`X]{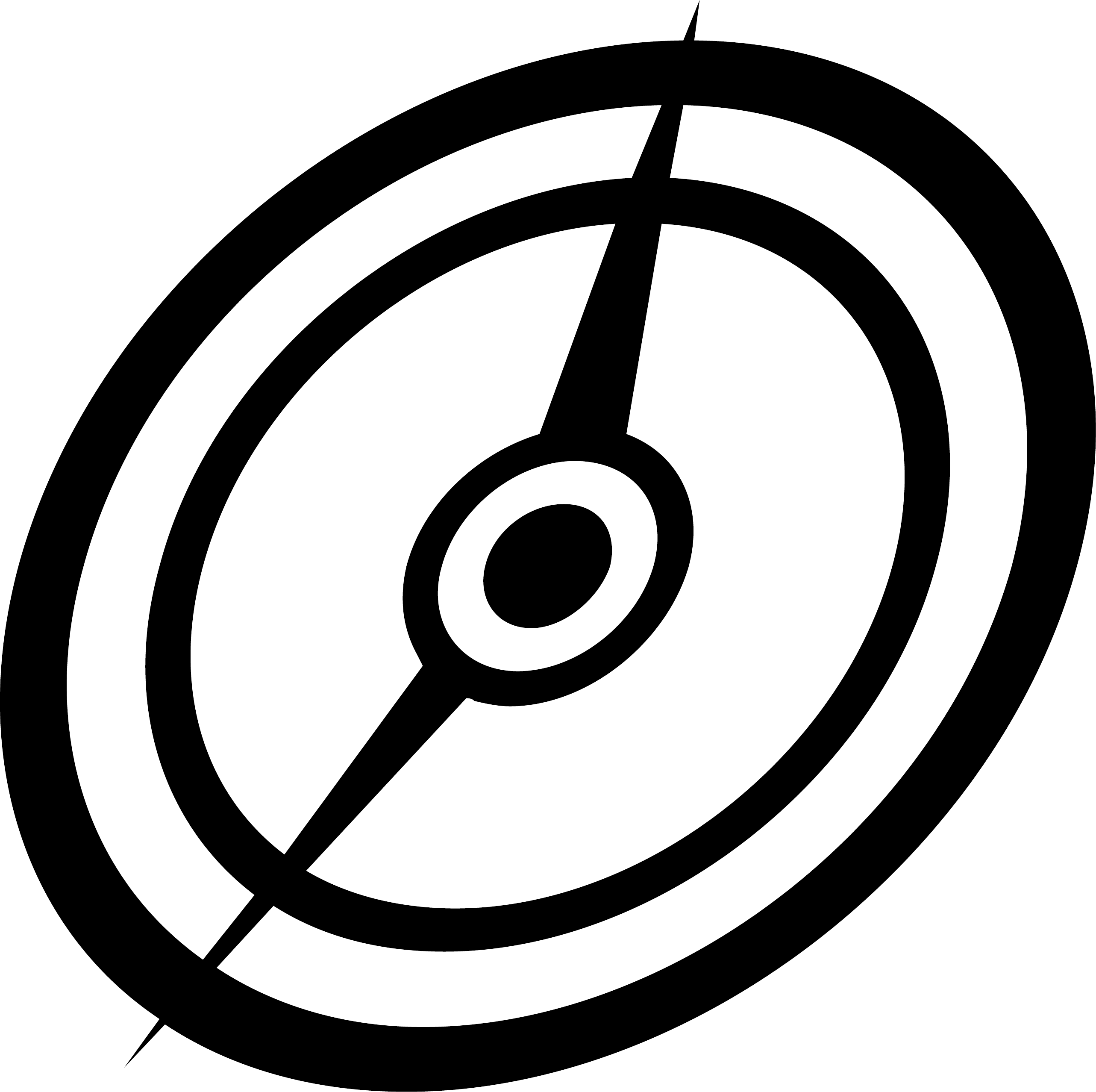}}
\newcommand{\trash}{\includegraphics[height=\fontcharht\font`X]{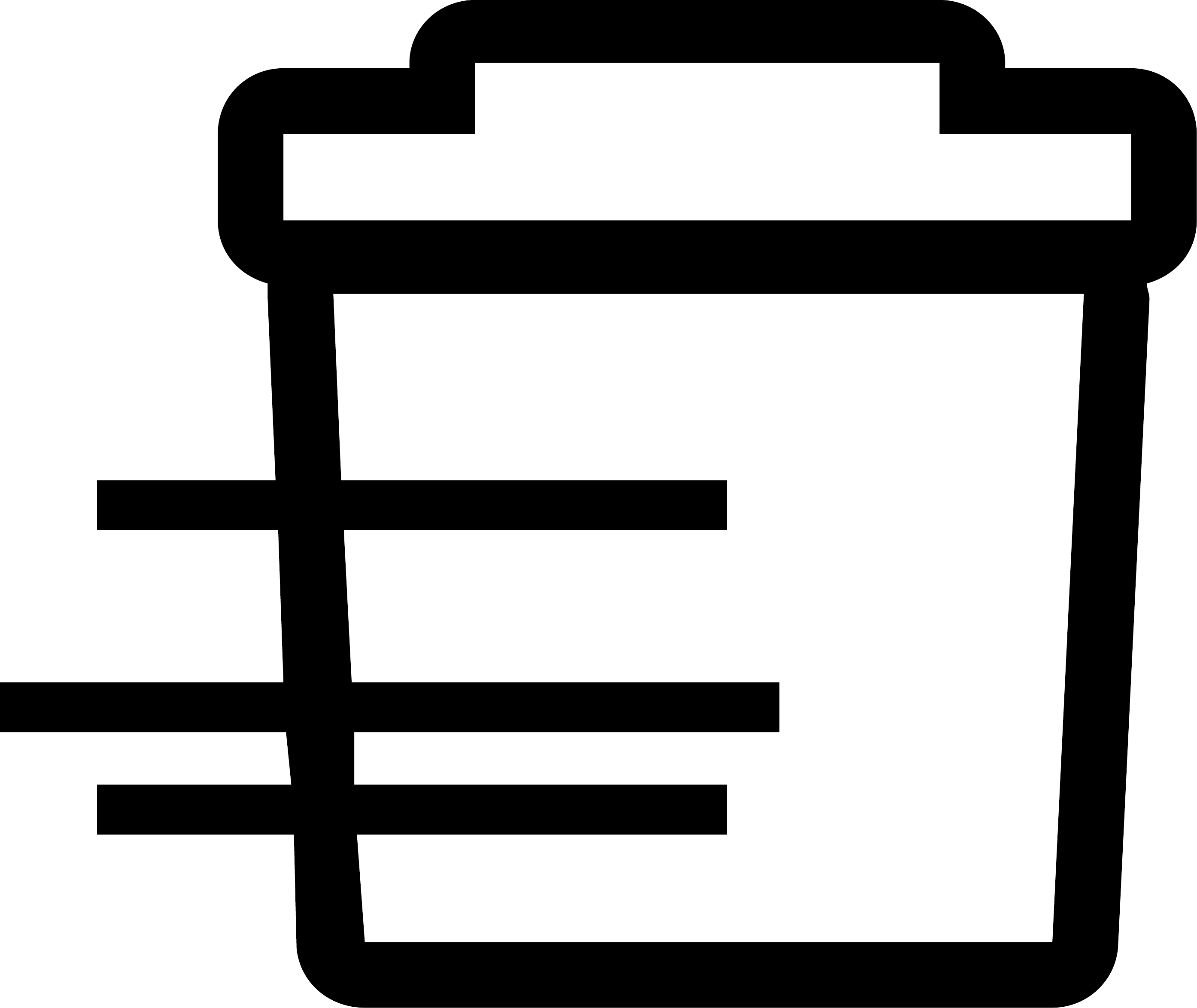}}

\colorlet{Weyland}{DarkGreen}
\colorlet{NBN}{DarkOrange}
\colorlet{Haas-Bioroid}{BlueViolet}
\colorlet{Jinteki}{Crimson}
\colorlet{Anarch}{OrangeRed}
\colorlet{Shaper}{LimeGreen}
\colorlet{Criminal}{RoyalBlue}
\definecolor{Adam}{RGB}{51,51,51}
\colorlet{Neutral}{Grey}

\newcommand{\cardbox}[9]{\begin{tcolorbox}[title={\href{https://netrunnerdb.com/en/card/#8}{\textsc{\large #1}}\\#2 -- #3},
  colframe=#6,
  coltitle=black,
  colbacktitle=#6!30, colback=#6!10,
  halign=flush left, halign title=flush left, halign lower=flush right,
  lower separated=false,
  top=1mm, middle=0mm
]#4\tcblower\textcolor{#6}{\myrepeat{#5}{$\bullet$}} #6 -- #7\end{tcolorbox}}


\begin{tcbraster}[raster columns=2,raster valign=top,size=small]

\input{netrunner-corp-cards}

\end{tcbraster}

\newpage\section{Netrunner Card Details: Runner}\label{sec:netrunner-cards2}
\vspace{-1em}
\begin{tcbraster}[raster columns=2,raster valign=top,size=small]

\input{netrunner-runner-cards}

\end{tcbraster}

\newpage
\bibliographystyle{plain}
\bibliography{cardgames}

\end{document}

%% file: netrunner-corp-board.pdf_tex
\begingroup%
  \makeatletter%
  \providecommand\color[2][]{%
    \errmessage{(Inkscape) Color is used for the text in Inkscape, but the package 'color.sty' is not loaded}%
    \renewcommand\color[2][]{}%
  }%
  \providecommand\transparent[1]{%
    \errmessage{(Inkscape) Transparency is used (non-zero) for the text in Inkscape, but the package 'transparent.sty' is not loaded}%
    \renewcommand\transparent[1]{}%
  }%
  \providecommand\rotatebox[2]{#2}%
  \ifx\svgwidth\undefined%
    \setlength{\unitlength}{540.72001648bp}%
    \ifx\svgscale\undefined%
      \relax%
    \else%
      \setlength{\unitlength}{\unitlength * \real{\svgscale}}%
    \fi%
  \else%
    \setlength{\unitlength}{\svgwidth}%
  \fi%
  \global\let\svgwidth\undefined%
  \global\let\svgscale\undefined%
  \makeatother%
  \begin{picture}(1,0.96671105)%
    \put(0,0){\includegraphics[width=\unitlength,page=1]{netrunner-corp-board.pdf}}%
    \put(0.04194407,0.93741678){\color[rgb]{1,1,1}\makebox(0,0)[lb]{\smash{Dedicated}}}%
    \put(0.04194407,0.91211718){\color[rgb]{1,1,1}\makebox(0,0)[lb]{\smash{Response}}}%
    \put(0.04194407,0.88681758){\color[rgb]{1,1,1}\makebox(0,0)[lb]{\smash{Team}}}%
    \put(0.01065246,0.68175766){\color[rgb]{0,0,0}\makebox(0,0)[lb]{\smash{\small{$|A|+3$ total Enigmas}}}}%
    \put(0,0){\includegraphics[width=\unitlength,page=2]{netrunner-corp-board.pdf}}%
    \put(0.01331558,0.64181094){\color[rgb]{0,0,0}\makebox(0,0)[lb]{\smash{Enigma}}}%
    \put(0.1584554,0.54194409){\makebox(0,0)[lb]{\smash{2}}}%
    \put(0,0){\includegraphics[width=\unitlength,page=3]{netrunner-corp-board.pdf}}%
    \put(0.03328895,0.61651134){\color[rgb]{0,0,0}\makebox(0,0)[lb]{\smash{Lose click, if able.}}}%
    \put(0,0){\includegraphics[width=\unitlength,page=4]{netrunner-corp-board.pdf}}%
    \put(0.03328895,0.59387485){\color[rgb]{0,0,0}\makebox(0,0)[lb]{\smash{End the run.}}}%
    \put(0,0){\includegraphics[width=\unitlength,page=5]{netrunner-corp-board.pdf}}%
    \put(0.01331558,0.4660453){\color[rgb]{0,0,0}\makebox(0,0)[lb]{\smash{Enigma}}}%
    \put(0.1584554,0.36617845){\makebox(0,0)[lb]{\smash{2}}}%
    \put(0,0){\includegraphics[width=\unitlength,page=6]{netrunner-corp-board.pdf}}%
    \put(0.03328895,0.44074569){\color[rgb]{0,0,0}\makebox(0,0)[lb]{\smash{Lose click, if able.}}}%
    \put(0,0){\includegraphics[width=\unitlength,page=7]{netrunner-corp-board.pdf}}%
    \put(0.03328895,0.41810921){\color[rgb]{0,0,0}\makebox(0,0)[lb]{\smash{End the run.}}}%
    \put(0,0){\includegraphics[width=\unitlength,page=8]{netrunner-corp-board.pdf}}%
    \put(0.01331558,0.29027965){\color[rgb]{0,0,0}\makebox(0,0)[lb]{\smash{Enigma}}}%
    \put(0.1584554,0.19041281){\makebox(0,0)[lb]{\smash{2}}}%
    \put(0,0){\includegraphics[width=\unitlength,page=9]{netrunner-corp-board.pdf}}%
    \put(0.03328895,0.26498005){\color[rgb]{0,0,0}\makebox(0,0)[lb]{\smash{Lose click, if able.}}}%
    \put(0,0){\includegraphics[width=\unitlength,page=10]{netrunner-corp-board.pdf}}%
    \put(0.03328895,0.24234357){\color[rgb]{0,0,0}\makebox(0,0)[lb]{\smash{End the run.}}}%
    \put(0,0){\includegraphics[width=\unitlength,page=11]{netrunner-corp-board.pdf}}%
    \put(0.01331558,0.11451405){\color[rgb]{1,1,1}\makebox(0,0)[lb]{\smash{Ice Wall}}}%
    \put(0,0){\includegraphics[width=\unitlength,page=12]{netrunner-corp-board.pdf}}%
    \put(0.03328895,0.08921445){\color[rgb]{1,1,1}\makebox(0,0)[lb]{\smash{End the run.}}}%
    \put(0.13448734,0.01731031){\color[rgb]{0,0,0}\makebox(0,0)[lb]{\smash{$3a_1{-}2$}}}%
    \put(0.08655126,0.02663117){\color[rgb]{0,0,0}\makebox(0,0)[lb]{\smash{$a_1$}}}%
    \put(0,0){\includegraphics[width=\unitlength,page=13]{netrunner-corp-board.pdf}}%
    \put(0.30825566,0.93741678){\color[rgb]{0,0,0}\makebox(0,0)[lb]{\smash{Priority}}}%
    \put(0.30825566,0.91211718){\color[rgb]{0,0,0}\makebox(0,0)[lb]{\smash{Requisition}}}%
    \put(0,0){\includegraphics[width=\unitlength,page=14]{netrunner-corp-board.pdf}}%
    \put(0.32157122,0.87083888){\color[rgb]{1,1,1}\makebox(0,0)[lb]{\smash{K. P. Lynn}}}%
    \put(0,0){\includegraphics[width=\unitlength,page=15]{netrunner-corp-board.pdf}}%
    \put(0.27962716,0.64181094){\color[rgb]{1,1,1}\makebox(0,0)[lb]{\smash{Ice Wall}}}%
    \put(0,0){\includegraphics[width=\unitlength,page=16]{netrunner-corp-board.pdf}}%
    \put(0.29960052,0.61651134){\color[rgb]{1,1,1}\makebox(0,0)[lb]{\smash{End the run.}}}%
    \put(0.40079891,0.54460719){\color[rgb]{0,0,0}\makebox(0,0)[lb]{\smash{$3a_2{-}2$}}}%
    \put(0.35286285,0.55392814){\color[rgb]{0,0,0}\makebox(0,0)[lb]{\smash{$a_2$}}}%
    \put(0,0){\includegraphics[width=\unitlength,page=17]{netrunner-corp-board.pdf}}%
    \put(0.27962716,0.4660453){\color[rgb]{1,1,1}\makebox(0,0)[lb]{\smash{Ice Wall}}}%
    \put(0,0){\includegraphics[width=\unitlength,page=18]{netrunner-corp-board.pdf}}%
    \put(0.29960052,0.44074569){\color[rgb]{1,1,1}\makebox(0,0)[lb]{\smash{End the run.}}}%
    \put(0.40079891,0.36884159){\color[rgb]{0,0,0}\makebox(0,0)[lb]{\smash{$3a_3{-}2$}}}%
    \put(0.35286285,0.37683096){\color[rgb]{0,0,0}\makebox(0,0)[lb]{\smash{$a_3$}}}%
    \put(0,0){\includegraphics[width=\unitlength,page=19]{netrunner-corp-board.pdf}}%
    \put(0.27962716,0.29027965){\color[rgb]{1,1,1}\makebox(0,0)[lb]{\smash{Ice Wall}}}%
    \put(0,0){\includegraphics[width=\unitlength,page=20]{netrunner-corp-board.pdf}}%
    \put(0.29960052,0.26498005){\color[rgb]{1,1,1}\makebox(0,0)[lb]{\smash{End the run.}}}%
    \put(0.39547268,0.19174432){\color[rgb]{0,0,0}\makebox(0,0)[lb]{\smash{$3a_{\frac{n}{2}}{-}2$}}}%
    \put(0.35286285,0.20239686){\color[rgb]{0,0,0}\makebox(0,0)[lb]{\smash{$a_{\frac{n}{2}}$}}}%
    \put(0,0){\includegraphics[width=\unitlength,page=21]{netrunner-corp-board.pdf}}%
    \put(0.27962716,0.11451405){\color[rgb]{1,1,1}\makebox(0,0)[lb]{\smash{Archer}}}%
    \put(0.42476697,0.01464721){\makebox(0,0)[lb]{\smash{6}}}%
    \put(0,0){\includegraphics[width=\unitlength,page=22]{netrunner-corp-board.pdf}}%
    \put(0.29960052,0.08921445){\color[rgb]{1,1,1}\makebox(0,0)[lb]{\smash{The Corp gains 2.}}}%
    \put(0,0){\includegraphics[width=\unitlength,page=23]{netrunner-corp-board.pdf}}%
    \put(0.29960052,0.06657797){\color[rgb]{1,1,1}\makebox(0,0)[lb]{\smash{End the run.}}}%
    \put(0,0){\includegraphics[width=\unitlength,page=24]{netrunner-corp-board.pdf}}%
    \put(0.57456721,0.93741678){\color[rgb]{1,1,1}\makebox(0,0)[lb]{\smash{Weyland (ID)}}}%
    \put(0,0){\includegraphics[width=\unitlength,page=25]{netrunner-corp-board.pdf}}%
    \put(0.58322237,0.88681758){\color[rgb]{1,1,1}\makebox(0,0)[lb]{\smash{Strongbox ($\times$2)}}}%
    \put(0,0){\includegraphics[width=\unitlength,page=26]{netrunner-corp-board.pdf}}%
    \put(0.54593873,0.64181094){\color[rgb]{1,1,1}\makebox(0,0)[lb]{\smash{Ice Wall}}}%
    \put(0,0){\includegraphics[width=\unitlength,page=27]{netrunner-corp-board.pdf}}%
    \put(0.5659121,0.61651134){\color[rgb]{1,1,1}\makebox(0,0)[lb]{\smash{End the run.}}}%
    \put(0.65912113,0.54460719){\color[rgb]{0,0,0}\makebox(0,0)[lb]{\smash{$3a_{n\scalebox{0.75}[1.0]{-}1}{-}2$}}}%
    \put(0.60585886,0.55392814){\color[rgb]{0,0,0}\makebox(0,0)[lb]{\smash{$a_{n-1}$}}}%
    \put(0,0){\includegraphics[width=\unitlength,page=28]{netrunner-corp-board.pdf}}%
    \put(0.54593873,0.4660453){\color[rgb]{1,1,1}\makebox(0,0)[lb]{\smash{Ice Wall}}}%
    \put(0,0){\includegraphics[width=\unitlength,page=29]{netrunner-corp-board.pdf}}%
    \put(0.5659121,0.44074569){\color[rgb]{1,1,1}\makebox(0,0)[lb]{\smash{End the run.}}}%
    \put(0.65912113,0.36884159){\color[rgb]{0,0,0}\makebox(0,0)[lb]{\smash{$3a_{n\scalebox{0.75}[1.0]{-}2}{-}2$}}}%
    \put(0.60585886,0.37683096){\color[rgb]{0,0,0}\makebox(0,0)[lb]{\smash{$a_{n-2}$}}}%
    \put(0,0){\includegraphics[width=\unitlength,page=30]{netrunner-corp-board.pdf}}%
    \put(0.54593873,0.29027965){\color[rgb]{1,1,1}\makebox(0,0)[lb]{\smash{Ice Wall}}}%
    \put(0,0){\includegraphics[width=\unitlength,page=31]{netrunner-corp-board.pdf}}%
    \put(0.5659121,0.26498005){\color[rgb]{1,1,1}\makebox(0,0)[lb]{\smash{End the run.}}}%
    \put(0.65778958,0.19174432){\color[rgb]{0,0,0}\makebox(0,0)[lb]{\smash{$3a_{\frac{n}{2}\scalebox{0.7}{+}1}\scalebox{0.75}[1.0]{$-$}2$}}}%
    \put(0.60585886,0.20239686){\color[rgb]{0,0,0}\makebox(0,0)[lb]{\smash{$a_{\frac{n}{2}+1}$}}}%
    \put(0,0){\includegraphics[width=\unitlength,page=32]{netrunner-corp-board.pdf}}%
    \put(0.54593873,0.11451405){\color[rgb]{1,1,1}\makebox(0,0)[lb]{\smash{Archer}}}%
    \put(0.69107855,0.01464721){\makebox(0,0)[lb]{\smash{6}}}%
    \put(0,0){\includegraphics[width=\unitlength,page=33]{netrunner-corp-board.pdf}}%
    \put(0.5659121,0.08921445){\color[rgb]{1,1,1}\makebox(0,0)[lb]{\smash{The Corp gains 2.}}}%
    \put(0,0){\includegraphics[width=\unitlength,page=34]{netrunner-corp-board.pdf}}%
    \put(0.5659121,0.06657797){\color[rgb]{1,1,1}\makebox(0,0)[lb]{\smash{End the run.}}}%
    \put(0,0){\includegraphics[width=\unitlength,page=35]{netrunner-corp-board.pdf}}%
    \put(0.87416774,0.92676431){\color[rgb]{0,0,0}\makebox(0,0)[lb]{\smash{R\&D}}}%
    \put(0.80958718,0.68175766){\color[rgb]{0,0,0}\makebox(0,0)[lb]{\smash{\small{$|A|+3$ total Enigmas}}}}%
    \put(0,0){\includegraphics[width=\unitlength,page=36]{netrunner-corp-board.pdf}}%
    \put(0.81225031,0.64181094){\color[rgb]{0,0,0}\makebox(0,0)[lb]{\smash{Enigma}}}%
    \put(0.95739012,0.54194409){\makebox(0,0)[lb]{\smash{2}}}%
    \put(0,0){\includegraphics[width=\unitlength,page=37]{netrunner-corp-board.pdf}}%
    \put(0.83222368,0.61651134){\color[rgb]{0,0,0}\makebox(0,0)[lb]{\smash{Lose click, if able.}}}%
    \put(0,0){\includegraphics[width=\unitlength,page=38]{netrunner-corp-board.pdf}}%
    \put(0.83222368,0.59387485){\color[rgb]{0,0,0}\makebox(0,0)[lb]{\smash{End the run.}}}%
    \put(0,0){\includegraphics[width=\unitlength,page=39]{netrunner-corp-board.pdf}}%
    \put(0.81225031,0.4660453){\color[rgb]{0,0,0}\makebox(0,0)[lb]{\smash{Enigma}}}%
    \put(0.95739012,0.36617845){\makebox(0,0)[lb]{\smash{2}}}%
    \put(0,0){\includegraphics[width=\unitlength,page=40]{netrunner-corp-board.pdf}}%
    \put(0.83222368,0.44074569){\color[rgb]{0,0,0}\makebox(0,0)[lb]{\smash{Lose click, if able.}}}%
    \put(0,0){\includegraphics[width=\unitlength,page=41]{netrunner-corp-board.pdf}}%
    \put(0.83222368,0.41810921){\color[rgb]{0,0,0}\makebox(0,0)[lb]{\smash{End the run.}}}%
    \put(0,0){\includegraphics[width=\unitlength,page=42]{netrunner-corp-board.pdf}}%
    \put(0.81225031,0.29027965){\color[rgb]{0,0,0}\makebox(0,0)[lb]{\smash{Enigma}}}%
    \put(0.95739012,0.19041281){\makebox(0,0)[lb]{\smash{2}}}%
    \put(0,0){\includegraphics[width=\unitlength,page=43]{netrunner-corp-board.pdf}}%
    \put(0.83222368,0.26498005){\color[rgb]{0,0,0}\makebox(0,0)[lb]{\smash{Lose click, if able.}}}%
    \put(0,0){\includegraphics[width=\unitlength,page=44]{netrunner-corp-board.pdf}}%
    \put(0.83222368,0.24234357){\color[rgb]{0,0,0}\makebox(0,0)[lb]{\smash{End the run.}}}%
    \put(0,0){\includegraphics[width=\unitlength,page=45]{netrunner-corp-board.pdf}}%
    \put(0.81225031,0.11451405){\color[rgb]{1,1,1}\makebox(0,0)[lb]{\smash{Ice Wall}}}%
    \put(0,0){\includegraphics[width=\unitlength,page=46]{netrunner-corp-board.pdf}}%
    \put(0.83222368,0.08921445){\color[rgb]{1,1,1}\makebox(0,0)[lb]{\smash{End the run.}}}%
    \put(0.93342211,0.01731031){\color[rgb]{0,0,0}\makebox(0,0)[lb]{\smash{$3a_n{-}2$}}}%
    \put(0.88548603,0.02663117){\color[rgb]{0,0,0}\makebox(0,0)[lb]{\smash{$a_n$}}}%
  \end{picture}%
\endgroup%

%% file: netrunner-corp-cards.tex
\cardbox{Archer}{ICE: Sentry - Destroyer}{Rez:~4, Str:~6}{%
As an additional cost to rez Archer, the Corp must forfeit an agenda.
\smallskip

\subroutine{} The Corp gains 2\credit{}.
\smallskip

\subroutine{} Trash 1 program.
\smallskip

\subroutine{} Trash 1 program.
\smallskip

\subroutine{} End the run.
}{2}{Weyland}{Revised Core Set}{20084}

\cardbox{Dedicated Response Team}{Asset: Hostile}{Rez:~2, Trash:~3}{%
If the Runner is tagged, Dedicated Response Team gains "Whenever a successful run ends, do 2 meat damage."
}{3}{Weyland}{Revised Core Set}{20081}

\cardbox{Enigma}{ICE: Code Gate}{Rez:~3, Str:~2}{%
\subroutine{} The Runner loses \click{}, if able.
\smallskip

\subroutine{} End the run.
}{0}{Neutral}{Revised Core Set}{20129}

\cardbox{Fast Track}{Operation}{Cost:~0}{%
Search R\&D for an agenda, reveal it, and add it to HQ. Shuffle R\&D.
}{0}{Neutral}{Honor and Profit}{05027}

\cardbox{Hedge Fund}{Operation: Transaction}{Cost:~5}{%
Gain 9\credit{}.
}{0}{Neutral}{Revised Core Set}{20132}

\cardbox{Ice Wall}{ICE: Barrier}{Rez:~1, Str:~1}{%
Ice Wall can be advanced and has +1 strength for each advancement token on it.
\smallskip

\subroutine{} End the run.
}{1}{Weyland}{Revised Core Set}{20088}

\cardbox{$\blackdiamond$K. P. Lynn}{Upgrade: Executive}{Rez:~1, Trash:~3}{%
Whenever the Runner passes all of the ice protecting this server, he or she must either take 1 tag or end the run.
}{2}{Weyland}{Terminal Directive}{13052}

\cardbox{Mandatory Seed Replacement}{Agenda: Security}{4/2}{%
When you score Mandatory Seed Replacement, rearrange any number of ice protecting all servers.
}{0}{Jinteki}{Free Mars}{12092}

\cardbox{Medical Breakthrough}{Agenda: Research}{4/2}{%
Lower the advancement requirement of each Medical Breakthrough by 1. This ability is active even while Medical Breakthrough is in the Runner's score area.
}{0}{Jinteki}{Honor and Profit}{05005}

\cardbox{Nisei Division: The Next Generation}{Identity: Division}{Deck:~45, Influence:~15}{%
Whenever you and the Runner reveal secretly spent credits, gain 1\credit{}.
}{0}{Jinteki}{Honor and Profit}{05002}

\cardbox{Priority Requisition}{Agenda: Security}{5/3}{%
When you score Priority Requisition, you may rez a piece of ice ignoring all costs.
}{0}{Neutral}{Revised Core Set}{20125}

\cardbox{Strongbox}{Upgrade}{Rez:~3, Trash:~1}{%
Each time the Runner accesses an agenda from this server, he or she must spend \click{} as an additional cost in order to steal it. This applies even during the run on which the Runner trashes Strongbox.
}{2}{Haas-Bioroid}{Revised Core Set}{20076}

\cardbox{Sub Boost}{Operation: Condition}{Cost:~0}{%
Install Sub Boost on a rezzed piece of ice as a hosted condition counter with the text "Host ice gains \textbf{barrier} and '\subroutine{} End the run.' after all its other subroutines."
}{0}{Neutral}{Order and Chaos}{07025}

\cardbox{Wall of Static}{ICE: Barrier}{Rez:~3, Str:~3}{%
\subroutine{} End the run.
}{0}{Neutral}{Revised Core Set}{20131}

\cardbox{Weyland Consortium: Building a Better World}{Identity: Megacorp}{Deck:~45, Influence:~15}{%
Gain 1\credit{} whenever you play a \textbf{transaction} operation.
}{0}{Weyland}{Revised Core Set}{20077}

%% file: netrunner-runner-cards.tex
\cardbox{Aurora}{Program: Icebreaker - Fracter}{Cost:~3, Mem:~1, Str:~1}{%
2\credit{}: Break \textbf{barrier} subroutine.
\smallskip

2\credit{}: +3 strength.
}{1}{Criminal}{Revised Core Set}{20027}

\cardbox{Escher}{Event: Run}{Cost:~3}{%
Make a run on HQ. If successful, instead of accessing cards, rearrange any number of ice protecting all servers (without rezzing or derezzing the ice). The same number of ice must be protecting each server after the rearrangement as before.
}{5}{Shaper}{Creation and Control}{03031}

\cardbox{Exile: Streethawk}{Identity: Natural}{Deck:~45, Influence:~15, Link:~1}{%
Whenever you install a program from your heap, draw 1 card.
}{0}{Shaper}{Creation and Control}{03030}

\cardbox{Grappling Hook}{Program}{Cost:~2, Mem:~1}{%
\trash{}: Break all but 1 subroutine on a piece of ice.
}{2}{Criminal}{Honor and Profit}{05045}

\cardbox{Infiltration}{Event}{Cost:~0}{%
Gain 2\credit{} or expose 1 card.
}{0}{Neutral}{Revised Core Set}{20055}

\cardbox{Pheromones}{Program: Virus}{Cost:~2, Mem:~1}{%
X\recurcredit{}
\smallskip

Use these credits during runs on HQ. X is the number of virus counters on Pheromones.
\smallskip

Whenever you make a successful run on HQ, place 1 virus counter on Pheromones.
}{2}{Criminal}{Revised Core Set}{20031}

\cardbox{$\blackdiamond$The Shadow Net}{Resource: Virtual}{Cost:~0}{%
\click{}\textbf{, forfeit an agenda:} Play an event from your heap, ignoring all costs.
}{0}{Neutral}{Terminal Directive}{13027}

%% file: arxiv-netrunner.bbl
\begin{thebibliography}{1}

\bibitem{MtGTuring}
Alex Churchill.
\newblock {M}agic: the {G}athering is {T}uring complete.
\newblock \url{http://www.toothycat.net/~hologram/Turing/}.

\bibitem{Phutball}
Erik~D. Demaine, Martin~L. Demaine, and David Eppstein.
\newblock Phutball endgames are hard.
\newblock In R.~J. Nowakowski, editor, {\em More Games of No Chance}, pages
  351--360. Cambridge University Press, 2002.
\newblock Collection of papers from the MSRI Combinatorial Game Theory Research
  Workshop, Berkeley, California, July 24--28, 2000.

\bibitem{Checkers}
A.~S. Fraenkel, M.~R. Garey, D.~S. Johnson, T.~Schaefer, and Y.~Yesha.
\newblock The complexity of checkers on an $n \times n$ board.
\newblock In {\em 19th Annual Symposium on Foundations of Computer Science
  (sfcs 1978)}, pages 55--64, Oct 1978.

\bibitem{GareyJohnson}
Michael~R. Garey and David~S. Johnson.
\newblock {\em Computers and intractability: a guide to the theory of
  NP-completeness}.
\newblock Series of books in the mathematical sciences. W. H. Freeman, 1979.

\bibitem{GPC}
Robert~A. Hearn and Erik~D. Demaine.
\newblock {\em Games, {P}uzzles, and {C}omputation}.
\newblock CRC Press, 2009.

\end{thebibliography}
